\documentclass[conference]{IEEEtran}
\usepackage[utf8]{inputenc}
\usepackage[T1]{fontenc}
\usepackage{mathtools}
\usepackage[thinc]{esdiff}
\usepackage{tabularx}
\usepackage{amsthm}
\usepackage{amsmath,amsfonts,amssymb}
\usepackage{algorithm}
\usepackage[noend]{algpseudocode}
\usepackage{float}
\usepackage[thinc]{esdiff}
\usepackage{epsf, epsfig, graphicx}
\usepackage{subcaption}
\usepackage{multirow}
\usepackage{cite}
\usepackage[english]{babel}
\usepackage{svg}

\DeclareGraphicsExtensions{.pdf, .bmp, .eps, .ps, .jpeg, .png, .svg}
\usepackage{calc}
\usepackage{booktabs}
\usepackage{array}
\usepackage{fourier} 
\long\def\comment#1{}

\usepackage{diagbox}

\newtheorem{theorem}{Theorem}

\newtheorem{corollary}[theorem]{Corollary}
\newtheorem{definition}{Definition}

\DeclareMathOperator*{\argmax}{arg\,max}
\DeclareMathOperator*{\argmin}{arg\,min}

\title{Fair Allocation of Bandwidth At Edge Servers For Concurrent Hierarchical Federated Learning

\vspace*{0.03in}}

\author{
    \IEEEauthorblockN{Md Anwar Hossen\textsuperscript{*}, Fatema Siddika\textsuperscript{*}, Wensheng Zhang}
    \IEEEauthorblockA{
        \textit{Department of Computer Science} \\
        \textit{Iowa State University}, Ames, USA \\
        \{manwar, fatemask, wzhang\}@iastate.edu}
    \vspace*{0.02in} 
    \IEEEauthorblockA{\textsuperscript{*}These authors have contributed equally to this work.}
}

\begin{document}
\maketitle

\begin{abstract}
This paper explores concurrent FL processes 
within a three-tier system, 
with edge servers between edge devices and FL servers. 
A challenge in this setup is  
the limited 
bandwidth from edge devices to edge servers. 
Thus, allocating the bandwidth efficiently and fairly to support simultaneous FL processes becomes crucial. 
We propose a game-theoretic approach to model the bandwidth allocation problem 
and develop 
distributed and centralized heuristic schemes to 
find an approximate Nash Equilibrium of the game. 
Through rigorous analysis and experimentation, 
we demonstrate that our schemes efficiently and fairly assign the 
bandwidth to the FL processes 
and outperform a baseline scheme where each edge server assigns bandwidth proportionally to the FL servers' requests that it receives.
The proposed distributed and centralized schemes have similar performance.

\end{abstract}

{\noindent\bf Keywords:} Federated Learning, Edge-Cloud Computing, Game Theory, Bandwidth Allocation.

\section{Introduction}

A substantial amount of training data from multiple sources 
is important to develop top-notch machine learning models~\cite{han2022tiff}.
Due to the proliferation of smartphones and devices, wearable devices, and IoT sensors,
a significant portion of data produced these days is acquired beyond the cloud,
particularly at the distributed end devices at the edge. 
Hence, much valuable real-world data exists in a highly distributed manner,
which could be utilized to train deep-learning models and develop more precise and intelligent applications.
Even though these smart devices collectively hold a substantial amount of valuable data,
their sensitive nature and diverse constraints  
- such as privacy concerns of their owners and privacy laws
- make it extremely difficult 
for companies to utilize or integrate such data without explicit consent for a special purpose.
As a result, methods that retain data on the local device 
while sharing the model have become more appealing.
Traditional centralized learning needs all data to be gathered and
stored in a central location on a data center or cloud server,
which raises the concern of privacy risks and data leakage. 
To overcome the challenges mentioned above,
federated learning (FL) has emerged as a new paradigm of distributed machine learning that
orchestrates model training across mobile devices~\cite{mcmahan2017communication}.

FL has taken an important and even faster-growing role in
the trend of decentralizing machine learning and
making learning more privacy-preserving.
The distributed nature has benefited from
the cloud-edge computing paradigm
that has already been widely adopted in the industry.

Earlier research on FL has typically assumed
a two-tier architecture,
i.e., a central server and multiple end devices.
In this setup, the central server can be a cloud server located deep within the Internet
or an edge server positioned at the network edge.
While a cloud server can access a vast number of end devices and
train deep learning models using massive datasets,
the communication between the server and the end devices can be slow and unpredictable
due to long and dynamic network connections.
Conversely, an edge server can only connect with nearby end devices and
utilize their limited datasets, but
the communication between the server and the end devices is more efficient
due to short-range communication and low latency \cite{liu2020client}.

The above tradeoff has further motivated the proposal of three-tier architectures
such as the client-edge-cloud hierarchical system~\cite{liu2020client},
where end devices form a large base of clients,
a cloud server acts as the central parameter server for the FL process,
and a number of edge servers lie in the middle connecting the above two tiers.
This emerging architecture is expected to combine the advantages of
both the cloud server based and the edge server-based two-tier architecture,
and to deliver good scalability and performance at the same time.

\begin{figure}
    \centering
    \includegraphics[width=8cm]{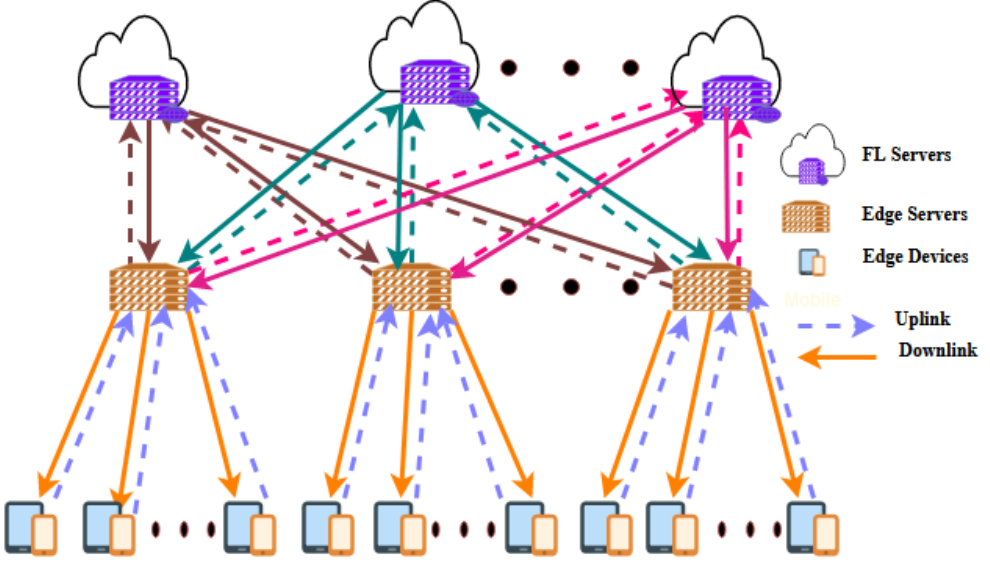}
    \caption{Proposed Three-Tier FL Architecture.}
    \label{fig:FL_Archi}
\end{figure}

As FL becomes increasingly popular, multiple FL processes could be ongoing simultaneously. Hence, in this paper, we extend the client-edge-cloud hierarchy into a more general three-tier architecture, as shown in Figure~\ref{fig:FL_Archi}, where multiple cloud-based FL servers could be at the top tier. The middle-tier edge servers should service the resource requests from the multiple ongoing FL processes led respectively by the multiple FL servers in the cloud. For example, the edge servers should allocate their communication bandwidth to connect the FL servers to their clients, which are end devices at the edge.
Ideally, each FL server should be able to communicate with as many of its clients as possible for fast convergence in learning. However, edge servers only have limited bandwidth and thus may not be able to simultaneously service all the requests from all the FL servers. Hence, it is a research problem how the middle-tier edge servers should allocate their bandwidth to achieve the following two goals: the bandwidth is utilized as fully as possible, and the allocation of bandwidth to cloud servers is as fair as possible. Note that the allocation decision should be made in a distributed manner because the involved parties have their own interests, and it is difficult to obtain their global state. 

To address the above research problem, we propose a game-theoretic approach. Specifically, we define a game to model  the interactions among the three-tier components. Each edge server owns a certain amount of bandwidth and sells them at a certain price to maximize its profit, which is the product of the sold bandwidth units and the price. Meanwhile, each cloud server has a certain amount of funds and uses the fund to purchase bandwidth from all edge servers it connects with so that it can communicate with as many end devices as possible. We formalize each participant's strategies, utilities, and constraints and define and solve the Nash Equilibrium, with which the allocation of the edge servers' bandwidth is optimal because the bandwidth is fully utilized and fairly distributed to cloud servers based on their funds. In addition to the above theoretical study, we also developed a distributed scheme that allows the edge servers and cloud servers to find an allocation of the bandwidth with high utility and fair allocation without any central point of control. Due to the distributed nature, such a scheme should be feasible for practical application.
\par Extensive experiments have been conducted to evaluate the performance of the proposed distributed scheme. The experiments consider various scenarios with heterogeneous conditions and compare to two reference schemes, i.e., the centralized scheme with a hypothetical central point of control and a baseline (distributed) scheme where each edge server simply allocates its bandwidth to requesting cloud servers proportionally to their request sizes. The results show that our proposed distributed scheme performs similarly to the centralized scheme but outperforms the baseline scheme.
A substantial amount of training data from diverse sources is vital for advanced machine learning models~\cite{han2022tiff}. As data is increasingly generated outside the cloud from smartphones, wearables, and IoT sensors, this valuable decentralized data can enhance deep-learning models for more precise and intelligent applications. However, privacy concerns and laws make it difficult for companies to use or integrate user data without explicit consent. To address this, Federated Learning (FL) emerged as a new approach, orchestrating model training across edge devices~\cite{mcmahan2017communication}.
Due to the  prominent features, many FL processes are commonly ongoing at any time. However, FL faces many challenges, including data heterogeneity, communication efficiency, and privacy preservation.
This paper focuses on addressing the challenge of communication efficiency, especially the uplink bandwidth.
Furthermore, the limited communication resources on edge devices, where FL clients locate, pose additional challenges to communication efficiency. When many FL processes are ongoing
concurrently, it further emerges as a new problem as how do we coordinate the simultaneous tasks regarding resource scheduling among the FL clients involved in these FL processes. In networked systems nowadays, edge servers can play an important role in coordination and optimization
as they connect servers and their overlapping pools of clients. Research on FL has typically
assumed a hierarchical architecture where the central server is in the top tier and edge devices are in the bottom tier. In this architecture, a central server can
be a cloud server located deep within the Internet or an
edge device positioned at the network edge. While a cloud server can access many edge devices and train deep learning models using a larger dataset, the communication between the server and the edge devices can be slow and unpredictable due to long and dynamic network
connections. Conversely, communication between the edge server and the edge
devices are more efficient due to short-range communication
and low latency and utilize their limited datasets \cite{liu2020client}. This situation motivates us to trade-offs between cloud and edge server-based architectures, cloud servers with access to extensive training data samples, and edge servers, enabling quick model updates in low latency. 

This paper has proposed a new multi-tier hierarchical architecture in Figure \ref{fig:FL_Archi} to address the uplink bandwidth allocation problem and provided detailed game-theoretic formulation, design, experiments, and system evaluation. The proposed hierarchical structure has multiple FL servers at the top,  multiple edge servers in the middle, and edge devices at the bottom. 
The concurrent operation of multiple federated learning servers improves scalability, load balancing, and fault tolerance and enables geographically distributed collaborations, resulting in enhanced efficiency and resilience of federated learning processes. 

In this paper, a virtual market is built to model and address the bandwidth allocation problem where edge servers compete by selling the bandwidth and making a profit. 
The FL servers compete to acquire as much uplink bandwidth from the edge servers with certain available funds.  This competitive behavior can be seen in an oligopolistic market in which the edge servers and the FL servers are motivated to act based on their self-interests instead of social benefit. In this scenario, we model the interaction between the FL and edge servers to address the bandwidth scheduling problem as a Stackleberg Game~\cite{pu4597505}. 

The major contributions we made through this paper are summarized as follows: 
\begin{itemize}
    \item 
    We define a three-tier hierarchical architecture for concurrent FL systems, where each FL server leads an FL learning process.
    \item 
    To solve the uplink bandwidth allocation problem, 
    we propose a game-theoretic scheme where 
    the interactions among the FL servers and the edge servers 
    are modeled as a Stackleberg game.
    \item 

    We design one centralized and another fully distributed heuristic schemes, to approximately solve the Nash Equilibrium of the Stackleberg game, for the goals of efficient and fair allocation of uplink bandwidth at edge servers for concurrent FL processes. 
    \item 
    The results show that our proposed scheme 
    can fairly distribute uplink bandwidth (in both partial and fully distributed scenarios) in uniform and heterogeneous settings as long as the heterogeneity in resource distribution is not too extreme. 
\end{itemize}

\par To the best of our knowledge, this is the first work that studies the problem of bandwidth allocation at edge servers to support concurrent FL processes in a fully decentralized (distributed) approach.  
\par In the rest of the paper, Section II describes the problem, and Section III presents a game-theoretic formulation of the problem, which is followed by a design of a heuristic scheme in Section IV. Section V reports the experiment results. Related works are briefly discussed in Section VI. Finally, Section VII concludes the paper.

\section{Problem Description}
\label{sec:problem}

\subsection{System Model}

We consider a system composed of three tiers of entities: 
{\em FL servers}, 
{\em edge servers} and 
{\em FL clients}. 
We denote the FL servers as $S_i$ for $i=0,\cdots,s-1$, where 
$s$ is the total number of FL servers, and
let ${\cal S}=\{S_0,\cdots,S_{s-1}\}$.
We denote the edge servers as $E_j$ for $j=0,\cdots,e-1$, where
$e$ is the total number of edge servers, and 
let ${\cal E}=\{E_0,\cdots,E_{e-1}\}$.
Each $S_i$ connects with each $E_j$; this way,
FL servers can utilize edge servers to efficiently interact with their clients, in particular, edge servers can aggregate the model updates from FL clients.
Each $E_j$ connects to $c_j$ edge devices which can work as FL clients for certain FL process, 
while each client only connects to one edge server.
We let $c = \sum_{j=0}^{e-1}c_{j}$ denote the total number of clients in the system.  

Regarding the communication between $E_j$ and its connected edge devices, we assume that
the uplinks (i.e., the communication links from the edge devices to the edge servers) 
have smaller bandwidth than the downlinks (i.e., 
the communication links from the edge server to the edge devices),
which is common for edge networks.
Also, for federated learning, the downlinks are used only by the edge server 
to multicast updated global model while 
the uplinks are used to unicast individual clients' model updates. 
Hence, we focus on the bandwidth allocation for uplinks and 
assume each $E_j$ has limited units of uplink bandwidth denoted as $b_j$.

When a device works as an FL client for FL process $i$
(i.e., the FL process launched by FL server $S_i$),
it requires certain units of bandwidth to upload its model update to the FL server $S_i$ via an Edge server.
We use $u_i$ to denote the bandwidth units required for this purpose.
By default, we assigned $u_i=1$. 

\subsection{Federated Learning Processes}

Each $S_i$ leads an FL process $i$ and acts as the central aggregation server. 
The set of potential clients for FL process $i$ is 
a subset of all the $c$ edge devices denoted as ${\cal C}_i$.
We denote the devices in ${\cal C}_i$ that are connected to each edge server $E_j$
as $C^{i}_{j,k}$ for $k=0,\cdots,c_{i,j}-1$, where
$c_{i,j}$ denotes the total number of such devices.

Each FL process $i$ follows the model of federated averaging (FedAvg)~\cite{mcmahan2017communication} with partial client participation,
where the selection of participating clients is determined by
$S_i$ following the strategy of random selection, specified in the following. 
First, $S_i$ determines a system parameter $\rho$,
the percentage of potential clients it wants to select.
Second, $S_i$ determines the weight of each of the potential clients.
Let each potential client $C^{i}_{j,k}$  
have a set of data samples denoted as $D^{i}_{j,k}$, 
whose size is denoted as $d^{i}_{j,k}=|D^{i}_{j,k}|$.
Each client $C^{i}_{j,k}$ should report $d^{i}_{j,k}$ to $S_i$.
Hence, $S_i$ can compute 
$d_{i}=\sum_{j=0}^{e-1}\sum_{k=0}^{c_{i,j}}d^{i}_{j,k}$ 
as the total size of data samples, and 
$\rho^{i}_{j,k} = \frac{d^{i}_{j,k}}{d^{i}}$ 
as the weight for selecting each client $C^{i}_{j,k}$.

Third, 
for each round, $S_i$ selects participating clients based on the above parameter and weights; 
more specifically, each client $C^{i}_{j,k}$ in ${\cal C}_i$ has the probability of $\rho\cdot\rho^{i}_{j,k}$ to be selected.
Thus, the expected number of such clients, 
is $\bar{c}^{i} = \rho\cdot|{\cal C}_i|$.
We let $\rho=1.0$ by default and thus
$\bar{c}^{i} = |{\cal C}_i|$.

\subsection{Design Goals}

This research has two main design goals.

(i) {\em Bandwidth Utility}: We aim to fully utilize the edge servers' available bandwidth; 
this way, the FL processes can complete as soon as possible. 

(ii) {\em Fairness}: We aim to have each FL processes receive a fair share of the available bandwidth.
Here, fairness has following implications:
for FL servers who invest the same amount of fund, 
they should receive the same or similar amount of bandwidth;
for FL servers who invest different amounts of fund,
they should receive the bandwidth proportionally to their funds.

\section{Game-Theoretic Formulation and Analysis}
\label{sec:scheme}



We apply the game-theoretic approach to formulate and address the afore-described problem. 
Intuitively, the interactions among the FL and the edge servers can be modeled as follows:
Each edge server $E_j$ sells its available uplink bandwidth 
at a unit price denoted as $p_j$.
Each FL server $S_i$ buys certain units of uplink bandwidth from each $E_j$,
and the bandwidth would allow a subset of the clients connected to $E_j$ 
to participate in the FL process $i$. 
The goal of each $E_j$ is to maximize its profit, 
which is the product of the sold units of bandwidth and the unit price.
The goal of each $S_i$ is to maximize its profit, 
which is the percentage of its selected clients that have the bandwidth purchased to 
participate in the FL process $i$. 

All the parties in the game have the aligned goal of 
utilizing as much available bandwidth as possible,
but they may conflict on bandwidth assignment.
If all the FL servers have equal funds for purchasing bandwidth,
the solution of the game provides an allocation that can achieve 
both the bandwidth utility 
and fair allocation of bandwidth; 
if the FL servers have different amounts of fund,
the solution of the game provides an allocation that achieves bandwidth utility 
and priority-based fairness, where 
the priority of a FL server is determined by its amount of funds.

\subsection{Assumption: Competing System}

The bandwidth allocation problem 
is challenging mainly because
FL servers compete for bandwidth at edge servers.
That is, the limited bandwidth available at the edge servers may be requested simultaneously by
multiple FL servers to communicate with their FL clients respectively, 
Such conflicts should be resolved appropriately
to achieve efficient and fair sharing of bandwidth, 
so as to support concurrent FL processes effectively. 
To formulate the competition, 
we introduce the concept of {\em competing system} as follows. 

A system $({\cal S}, {\cal E})$, 
including ${\cal S}$ of $s$ FL servers 
and ${\cal E}$ of $e$ edge servers,
is a competing system if:
    (i) For every $S_i,S_j\in{\cal S}$:
        there exists an edge server $E_w\in{\cal E}$ such that
        $S_i$ and $S_j$ both have at least one selected client connecting to $E_w$; or
        there exists another FL server $S_k\in{\cal S}$ such that,
        $S_i$ and $S_k$ are competing, and $S_j$ and $S_k$ are competing. 
    (ii) For every $E_w\in{\cal E}$, 
    there exist $S_i,S_j\in{\cal S}$ such that
    $S_i$ and $S_j$ both have at least one selected client connecting to $E_w$. 
%

In this paper,
we assume the system that we study is a completing system,
as this is the most challenging setting 
for the bandwidth allocation problem. 
Note that, if an FL server is not part of a competing system,
it is trivial to assign bandwidth to the server; i.e., for each edge server,
this FL server can simply be assigned with either its requested bandwidth
or the bandwidth affordable by the edge server, whichever is smaller.

\subsection{Rules and Strategies of The Game}

Formally, we model the interaction among the FL servers and edge servers as 
a Stackleberg Game~\cite{pu4597505}  
with multiple leaders (i.e., edge servers) and multiple followers (i.e., FL servers). 
Recall that, a Stacklerberg Game models a decision-making system 
that involves a leader and followers taking turns making choices. 
The leader is the first to select its strategy, 
while the followers observe the leader's strategy before making their own decisions. 
Assume $E_j$ has $b_j$ units of bandwidth to sell;
each $S_i$ has a fund of amount $f_i$ used to purchase bandwidth for its clients.   
The game has the following stages and results: 

    
    {\em Stage I}. 
    Each $E_j$ claims a unit price $p_j$ for its bandwidth. 
    Hence, its strategy is 
    \begin{equation}\label{eq:strategy-edge}
        p_j\in {\cal R}^{+},
    \end{equation}
    where ${\cal R}^{+}$ denotes the set of positive real numbers.
    
    {\em Stage II}.
    Each $S_i$ responses to each $E_j$ its request $r_{i,j}$,
    indicating its willingness to buy $r_{i,j}$ units of bandwidth from $E_j$.
    Hence, its strategy is a vector 
    \begin{equation}\label{eq:strategy-fs}
        \vec{r}_{i,\cdot}=(r_{i,0},\cdots,r_{i,e-1})\in\{0,\cdots,b_j\}^{e},
    \end{equation}
    where each $r_{i,j}\in\{0,\cdots,b_j\}$.
    Therefore, a strategy profile of the game can be denoted as
    \begin{equation}\label{eq:strategyprofile-game}
        (p_0,\cdots,p_{e-1};~\vec{r}_{0,\cdot},\cdots,\vec{r}_{s-1,\cdot}).
    \end{equation}
    
    {\em Result}.
    Each $S_i$ is granted with 
        \begin{equation}\label{eq:rule-result}
        x_{i,j} = 
        \lfloor
            r_{i,j}\cdot\min(1,\frac{b_j}{\sum_{i=0}^{s-1}r_{i,j}})
        \rfloor
        \end{equation}
    units of bandwidth from each $E_j$ at unit price $p_j$. Specifically:
        (i) For each $E_j$, if its received requests do not exceed its supply of bandwidth,
        i.e., $\sum_{i=0}^{s-1}r_{i,j}\leq b_{j}$,
        each $S_i$ is granted with $x_{i,j}=r_{i,j}$ units of bandwidth.    
        (ii) Otherwise, for the sake of fairness, 
        each $S_i$ is granted with 
        $x_{i,j}=\lfloor\frac{r_{i,j}\cdot b_j}{\sum_{i=0}^{s-1}r_{i,j}}\rfloor$
        units of bandwidth.    
    
    We denote the total units of bandwidth sold to $S_i$ as 
    \begin{equation}
        x_{i,\cdot} = \sum_{j=0}^{e-1}x_{i,j}
    \end{equation}
    and the total units of bandwidth sold by $E_j$ as
    \begin{equation}
        x_{\cdot,j} = \sum_{i=0}^{s-1}x_{i,j}
    \end{equation}



\subsection{Constraints}

A result of the game is feasible only if the following constraints are satisfied:
    (i) Each $x_{i,j}$ is a multiple of $u_i$. 
    Thus, the bandwidth purchased by $S_i$ from each $E_j$ 
    can be fully used. 
    Particularly, no fragment of bandwidth is wasted as 
    not large enough to connect a client. 
    (ii) $x_{i,j}\leq u_i\cdot\bar{c}^i_j$. That is, 
    the bandwidth purchased by $S_i$ from $E_j$ 
    cannot exceed the total bandwidth needed by 
    all clients of $S_i$ connected to $E_j$. 
    (iii) $x_{\cdot,j}\leq b_j$ at each $E_j$. That is, the bandwidth sold by $E_j$ should not exceed the total bandwidth it owns.
    (iv) For each $S_i$,
    $\sum_{j=0}^{e-1}(p_j\cdot x_{i,j})\leq f_i$. That is, that $S_i$ spends to purchase bandwidth should not exceed $f_i$.

\subsection{Utility Functions}

The goal of each edge server is to gain high income from selling its bandwidth,
and that of each FL server is to purchase as much bandwidth as possible 
for its clients using a given fund. 
Note that, in this game model, a FL server's utility does not factor in
the cost (i.e., the consumption of fund), because the FL server is
expected to purchase as much bandwidth using the available fund. 
Their utilities should quantify the degree of satisfying their goals; also, it is natural that
the marginal increase of utilities should be diminishing. 
Hence, we have the following definitions. 

\begin{definition}\label{def:edge-utility}
(Edge Server's Utility).
The utility of each edge server $E_j$ is 
$u_{E}(p_j,x_{\cdot,j}) = log(1+\frac{p_j\cdot x_{\cdot,j}}{\sum_{i=0}^{s-1}f_i})$.    
\end{definition}
%

\begin{definition}\label{def:fs-utility}
(FL Servers' Utility).
The utility of each FL server $S_i$ is
$u_{S}(x_{i,\cdot}) = log(1+\frac{x_{i,\cdot}}{u_i\cdot\bar{c}^i})$.    
\end{definition}
%

\subsection{Nash Equilibrium}

We first define Nash Equilibrium~\cite{osborne2004introduction} for the FL servers (i.e., the followers in the game),
given a strategy profile of the edge servers (i.e., the leaders in the game).
Then, we define Nash Equilibrium for the whole game.

\begin{definition}\label{def:fs-ne} 
(FL Servers' Nash Equilibrium).
Given a strategy profile  
$(p_0,\cdots,p_{e-1})$
of all edge servers,
a strategy profile 
\begin{equation}\label{eq:FL-sp}
(r^{*}_{0,0},\cdots,r^{*}_{0,e-1}), \cdots, (r^{*}_{s-1,0},\cdots,r^{*}_{s-1,e-1})
\end{equation}
of all FL servers is a Nash Equilibrium iff:
    for every FL server $S_i$ there is no any 
    $(r_{i,0},\cdots,r_{i,e-1})\not=(r^{*}_{i,0},\cdots,r^{*}_{i,e-1})$
    such that,
    letting 
    $x^{*}_{i,\cdot}=(x^{*}_{i,0},\cdots,x^{*}_{i,e-1})$
    be the result of the game 
    where the FL servers' strategy profile is as (\ref{eq:FL-sp}) and
    $x_{i,\cdot}=(x_{i,0},\cdots,x_{i,e-1})$
    be the result of the game
    where the FL servers' strategy profile is changed from 
    (\ref{eq:FL-sp}) by replacing 
    $r^{*}_{i,0},\cdots,r^{*}_{i,e-1}$ with 
    $r_{i,0},\cdots,r_{i,e-1}$, 
    $u_{S}(x_i) > u_{S}(x^{*}_i)$.
\end{definition}

\begin{definition}\label{def:game-ne}
(Game's Nash Equilibrium).
Strategy profile 
\begin{equation}\label{eq:sp-ne}
    p^{*}_0,\cdots,p^{*}_{e-1}; \vec{r^{*}}_{0,\cdot},\cdots,\vec{r^{*}}_{s-1,\cdot}
\end{equation}
is a Nash Equilibrium of the game iff:
    (1) Given the edge servers' strategy profile $(p^{*}_0,\cdots,p^{*}_{e-1})$,
    strategy profile $sp^{*}=(\vec{r^{*}}_{0,\cdot},\cdots,\vec{r^{*}}_{s-1,\cdot})$ 
    is the FL servers' Nash Equilibrium.
    %
    (2) For every edge server $E_j$ there is no $p_j\not=p^{*}_j$ such that:
            (2.1) there exists a FL servers' Nash Equilibrium $sp$ for 
            the edge servers' strategy profile
                \begin{equation}\label{eq:edge-sp-new}
                    (p^{*}_0,\cdots,p^{*}_{j-1},p_{j},p^{*}_{j+1},\cdots,p^{*}_{e-1});
                \end{equation}
            %
            (2.2) letting $x^{*}_{\cdot,j}$ be the game result with strategy profile as (\ref{eq:sp-ne}) and
            $x_{\cdot,j}$ be the game result with 
            the FL servers' strategy profile $sp$ and the edge servers' strategy profile as (\ref{eq:edge-sp-new}),
                \[
                u_{E}(p_j,x_{\cdot,j}) > u_{E}(p^{*}_j,x^{*}_{\cdot,j}).
                \]
\end{definition}

\subsection{Properties of Nash Equilibrium}

\begin{theorem}\label{theo:same-price}
    For any Nash Equilibrium of the game, as denoted by (\ref{eq:sp-ne}),
    it must hold that $p^*_0=\cdots=p^*_{e-1}$.
\end{theorem}





\begin{theorem}\label{theo:price-constraint}
    For any Nash Equilibrium of the game, as denoted by (\ref{eq:sp-ne}),
    where $p^*_0=\cdots=p^*_{e-1}=p$, the following must hold:
    for every FL server $S_i$, who has purchased $x^*_{i,\cdot}$ where $\frac{x^*_{i,\cdot}}{u_i}<\bar{c}^{i}$, in the result of the game,
    $(x^*_{i,\cdot}+u_i)\cdot p > f_i$. It cannot purchase more bandwidth to support one more client.  
\end{theorem}


\begin{theorem}\label{theo:price-minimiality}
    For any Nash Equilibrium of the game,
    the price offered by the edger servers should be the lowest one that meets the requirement specified in the previous theorem. 
\end{theorem}


Due to space limits,
we do not provide proof for the above theorems.
These theorems can be proved by contradiction based on the definitions of the competing systems and the game formulation. 

\subsection{Solution of Nash Equilibrium}

Based on the game formulation and analysis in the last section,
a Nash Equilibrium can be computed by 
solving the following optimization problem. 
\begin{eqnarray}
&&\min~p \label{eq:ipp-objective}\\
&&s.t. \nonumber\\
&&p>0\label{eq:ipp-positive-price}\\
&&x_{i,j}\leq \bar{c}^i_j\cdot u_i,~\forall~i,j\label{eq:ipp-c-1}\\
&&x_{\cdot,j}\leq b_j,~\forall~j\label{eq:ipp-c-2}\\
&&p\cdot x_{i,\cdot}\leq f_i \label{eq:ipp-c-3}\\
&&x_{i,\cdot}<\bar{c}^i\cdot u_i~\implies~f_i<p\cdot (x_{i,\cdot}+u_i),~\forall~i\label{eq:ipp-c-4}
\end{eqnarray}
Here, all the edge servers claim the same price denoted as $p$, 
due to Theorem~\ref{theo:same-price}. 
Constraints (\ref{eq:ipp-c-1}), (\ref{eq:ipp-c-2}) and (\ref{eq:ipp-c-3}) are due to the rules of the game.
Constraint (\ref{eq:ipp-c-4}) is due to Theorem~\ref{theo:price-constraint}. 
The objective (\ref{eq:ipp-objective}) is due to Theorem~\ref{theo:price-minimiality}.

\begin{theorem}\label{theo:solution-definition}
    The solution to optimization problem 
    defined by (\ref{eq:ipp-objective})-(\ref{eq:ipp-c-4})
    is a Nash Equilibrium of the game. 
\end{theorem}
\begin{proof} (sketch)
The proof is based on Definitions \ref{def:fs-ne} and \ref{def:game-ne} 
as well as Theorem \ref{theo:same-price}.

A solution to optimization problem (\ref{eq:ipp-objective})-(\ref{eq:ipp-c-4}),
$\langle p, \{x_{i,j}\}\rangle$, can be converted to 
a combination of a strategy profile
    \begin{equation}
    p^*_0,\cdots,p^*_{e-1}; \vec{r^*}_{0,\cdot},\cdots,\vec{r^*}_{s-1,\cdot}
    \end{equation}
based on the rules of the game.    
Here, $p^*_0=\cdots=p^*_{e-1}=p$, 
and each $x_{i,j}$ ($i\in\{0,\cdots,s-1\}$ and $j\in\{0,\cdots,e-1\}$) and $r^*_{i,j}$
satisfy equation (\ref{eq:rule-result}).

With proof by contradiction, 
it can be shown that the above strategy profile is a Nash Equilibrium of the game. 
First, when the strategy of the edge servers is $p^*_0=\cdots=p^*_{e-1}=p$,
none of the FL server $S_i$ can deviate from its strategy $r^*_{i,\cdot}$
due to the following reasons:
If the change of $S_i$ leads to an increase of 
its bandwidth assignment $x_{i,\cdot}$ by at least $u_i$,
this is impossible because this case implies that
$x_i<\vec{c}^i\cdot u_i$ and meanwhile $f_i\geq p(x_i+u_i)$,
which violates the condition of Equation (\ref{eq:ipp-c-4}).
Thus, no change can $S_i$ make to increase its utility. 
Second, when one edge server $e_j$ changes its strategy, 
from the above strategy profile,
i.e., the change makes $p_j\not p$ while $p_k=p$ for every $k\not=j$,
the resulting new strategy profile should not be a Nash Equilibrium
according to Theorem \ref{theo:same-price}.
\end{proof}

\begin{corollary}
    Nash Equilibrium exists for the game.
\end{corollary}
\begin{proof} (sketch)
According to Theorem \ref{theo:solution-definition},
a solution to the optimization problem (\ref{eq:ipp-objective})-(\ref{eq:ipp-c-4})
is a Nash Equilibrium; that is, 
if there is a solution to the optimization problem, the game should have Nash Equilibrium. 
Hence, we only need to show the optimization problem has a solution.

Given a set of parameters 
$
{\cal P}: 
\{u_i, \vec{c}^i_j, b_j, f_i\}
~\mbox{for}
~i\in\{0,\cdots,s-1\}
~\mbox{and}
~j\in\{0,\cdots,e-1\}$,
let $\phi_{\cal P}(p)$ be a function that returns 
$\{x_{i,j}|i\in\{0,\cdots,s-1;j\in\{0,\cdots,e-1\}$
that satisfy conditions (\ref{eq:ipp-positive-price})-(\ref{eq:ipp-c-4})
if such set of $\{x_{i,j}\}$ exists, or it returns $\emptyset$ otherwise. 
Hence, solving the optimization problem is equivalent to finding the minimum $p$
such that $\phi_{\cal P}(p)$ returns non-empty set.
Further, we can find at least one value for $p$ 
such that $\phi_{\cal P}(p)$ returns non-empty set.
For example, 
let $p=\max\{\frac{f_i}{u_i} | i=0,\cdots,s-1\}$,
$\phi_{\cal P}(p)$ returns 
$\{x_{i,j}=0|i\in\{0,\cdots,s-1;j\in\{0,\cdots,e-1\}$
This way, it is shown that a solution exists for the optimization problem. 
\end{proof}

\begin{algorithm}
\caption{Distributed Algorithm by Each FL Server $S_i$}
\label{alg:heuristic-distributive-fs}
{\bf Inputs:}
    (i) $u_i$ - number of bandwidth units needed by each client of $S_i$;
    (ii) $\bar{c}^i_j$ - expected number of clients 
    for $S_i$ connected to each $E_j$;
    (iii) $f_i$ - fund of $S_i$;
    (iv) $\Delta$ - minimum ratio between the lowest and highest
    prices of edge servers when the prices converge 
    ($0.9$ by default);
    (v) $\tau$ - step size for adjusting request (0.1 by default).


{\bf Variables and Initial Values:}
\begin{itemize}
    \item 
    $rr_{i,j}$: number of bandwidth units requested by $S_i$ 
    to each $E_j$, which is initialized to 
    $\lfloor \bar{c}^i_j \cdot u_i \cdot f_i\rfloor$
    \item $\text{IsConv}:$ indicator on if edge servers' prices converge, initialized to false
\end{itemize}

{\bf Algorithm:}
\begin{algorithmic}[1]
\While{$\text{!IsConv}$}
        \For{$j \gets 0,1,...$ to $e-1$}
            \State 
            send $rr_{i,j}$ to $E_j$
            \State wait till receive result 
            $(p_{j}, b_j, x_{i,j})$ 
        \EndFor
        \If{$ \frac{\min\{p_{j}\}}{\max\{p_{j}\} } > \Delta $}
                \State $\text{IsConv} \leftarrow true $
        \Else
            \State $\tilde{p} \leftarrow \frac
            {\sum_{j=0}^{e-1} (p_j \cdot b_j)}
            {\sum_{j=0}^{e-1} b_j}$
            \State $\delta{d_j} \leftarrow  
            \tilde{p} \cdot b_j - (p_j \cdot b_j)  
            ~\text{for}~j =0,\cdots,e-1$
            \For{$j \gets 0,1,...$ to $e-1$}
                \State $rr_{i,j}\leftarrow rr_{i,j} + 
                \delta d_j \cdot \tau$
            \EndFor 
        \EndIf
 \EndWhile

\State $\text{return } x_{i,j}$ for $j=0,\cdots,e-1$
\end{algorithmic}
\end{algorithm}

\section{Game-Theoretic Heuristic Schemes}

To solve the above optimization problem in a practical distributed system,
we propose a heuristic distributed scheme in 
Algorithms \ref{alg:heuristic-distributive-fs} 
and \ref{alg:heuristic-distributive-es},
executed by every FL server and 
every edge server, respectively.
Assuming a hypothetical central point of control,
we also propose a centralized scheme in 
Algorithm \ref{alg:heuristic-ne} as a reference for evaluation.

\begin{algorithm}
\caption{Distributed Algorithm by Each Edge Server $E_j$}
\label{alg:heuristic-distributive-es}
{\bf Inputs:}
    (i) $b_j$: available units of bandwidth at $E_j$;
    (ii) $s$: total number of FL servers connected to $E_j$. 

{\bf Upon receiving request $rr_{i,j}$ from $S_i$:} record the request. 

{\bf At the time to respond to FL servers' request (periodically every a predefined interval): }
\begin{algorithmic}[1]
    \For{$i \gets 0,1,...$ to $s-1$}
        \If{~new~$rr_{i,j}$~hasn't~been~received}
            \State apply linear regression to predict $rr_{i,j}$
        \EndIf
    \EndFor
    \State $ p_{j} \leftarrow {\frac{\sum_{i=0}^{s-1}rr_{i,j}}{b_{j}}}$ 
    \State \text{return}  $p_{j}, b_j, x_{i,j}$ for $i=0,\cdots,s-1$.
\end{algorithmic}
\end{algorithm}

\subsection{Distributed Scheme}

Algorithms \ref{alg:heuristic-distributive-fs} 
and \ref{alg:heuristic-distributive-es}
formally describes the interactions between each FL server $S_i$
and each edge server $E_j$ for every round of FL.

As specified in 
Algorithms~\ref{alg:heuristic-distributive-fs},
on the first-time interaction with $E_j$, 
$S_i$ sends the initial request computed as
$\lceil\bar{c}^i_j\cdot u_i\cdot f_i\rceil$,
which is the product of 
the total number of bandwidth units needed to support its selected clients (i.e., $\bar{c}^i_j\cdot u_i$) and
the amount of its fund (i.e., $f_i$).
While the inclusion of the first term is easy to understand,
the second term (i.e., $f_i$) is factored in to give higher
priority for FL server having large fund.
The request is sent to $E_j$.

As specified in 
Algorithm~\ref{alg:heuristic-distributive-es}, 
each edge server $E_j$ collects requests from 
all the FL servers that it connects with.
For every period of time, 
it allocates its available bandwidth to the FL servers
proportionally to their requests, and 
determines its current price as the ratio
between its received requests (i.e., demands) 
and its available bandwidth (i.e., supply). 
Note that, 
if it fails to receives a request from a FL server,
linear regression is applied to predict its request
for a number of periods. 

Upon receiving response 
(i.e., bandwidth allocation and current price) 
from $E_j$, as further specified in  
Algorithms~\ref{alg:heuristic-distributive-fs},
$S_i$ first checks if the received prices are close enough
(i.e., the ratio between the minimal and maximum prices
is greater than system parameter $\Delta$).
If so, the prices converge and the interactions complete
for this round. Otherwise, it
computes the system's overall bandwidth demands
(i.e., $\sum_{j=0}^{e-1}(p_j\cdot b_j)$) and 
overall bandwidth supply
(i.e., $\sum_{j=0}^{e-1}(b_j)$) to get 
the ideal price (i.e., demand/supply ratio) in Line 8. 
Based on the above,
it derives the ideal demands for every $E_j$ 
(i.e., $\tilde{p}\cdot b_j$),
and the difference between the ideal and current demands
(i.e., $\delta d_j$ in Line 9). 
Then, it adjusts its requests to every $E_j$ by 
$\delta d_j\cdot \tau$ where $\tau$ is a parameter specifying
the adjustment step, and sends the new requests.

\vspace*{-0.1in}
\subsection{Centralized Algorithm}
\vspace*{-0.05in}

In the centralized algorithm \ref{alg:heuristic-ne},
a rank is maintained for each FL server $S_i$, 
such that the lower the rank the higher the priority. 
The ranks for all FL servers are initialized to 0, and 
then the loop Line 1-12 is run until all FL servers' ranks reach $+\infty$.
For each iteration,
the FL server with the lowest rank value, 
whose id is $i^*$ as found in Line 2, 
is considered.
Then an appropriate edge server is selected to allocate $u_{i*}$ units of bandwidth for $S_{i^*}$,
which involves the following main steps.
As in Lines 3-4,
if none of the edge servers can provide $u_{i*}$ units of bandwidth,
the rank of $S_{i^*}$ is set to $+\infty$, and the iteration aborts.
Otherwise, an appropriate edge server, with identity $j^*$,
is selected in Lines 6-7.
The probability of selecting each $E_j$ is based on $f(i^*,j)$ in Line 6.
Specifically, the probability is determined by two factors:
the more bandwidth left at $E_j$ (i.e., $rb_j-u_{i^*}$), the higher is the probability; 
the more bandwidth is likely requested from $E_j$ by all FL servers (as estimated by the denominator of $f(i^*,j)$, the lower the probability.
Finally, after the bandwidth allocation has been determined,
the price is decided according to Theorem~\ref{theo:price-minimiality}.

\begin{algorithm}
\caption{Centralized Algorithm for Problem (\ref{eq:ipp-objective})-(\ref{eq:ipp-c-4})}
\label{alg:heuristic-ne}
{\bf Inputs:}
    (i) $u_i$: units of bandwidth needed by each client for FL $i$;
    (ii) $\bar{c}^i_j$: expected number of clients for $S_i$ connected to $E_j$;
    (iii) $b_j$: maximum units of bandwidth at $E_j$;
    (iv) $f_i$: fund at $S_i$

{\bf Outputs:}
\begin{itemize}
    \item 
    $p$: price per unit of bandwidth at every edge server
    \item 
    $x_{i,j}$: units of bandwidth allocated to $S_i$ by $E_j$
\end{itemize}

{\bf Variables and Initial Values:}
\begin{itemize}
    \item 
    $rb_j$: remaining (un-allocated) units of bandwidth at $E_j$
    \item 
    $rk_i$: rank for each $S_i$, which is initialized to $0$
    \item 
    $rr_{i,j}$: remaining request from each $S_i$ to each $E_j$
    which is initialized to $\min\{\lfloor\frac{rb_j}{u_i}\rfloor,\bar{c}^i_j\}\cdot u_i$
    \item 
    $x_{i,j}$: initialized to $0$
\end{itemize}
{\bf Algorithm:}
\begin{algorithmic}[1]
\While{$\min\{rk_i\}<+\infty$}
    \State $i^*~\leftarrow~\argmin_i~rk_i$
    \If{$\max\{rb_j\}<u_{i^*}$}
        $rk_{i^*}~\leftarrow~+\infty$
    \Else
        \State $f(i^*,j)=\frac{rb_j-u_{i^*}}{\sum_{rk_i<+\infty}{\frac{rr_{i,j}}{\sum_{j'}rr_{i,j'}}\cdot\frac{f_{i}}{\sum_{rk_i<+\infty}f_i}}}$
        \State $j^*\leftarrow\argmax_{(\forall j)x_{i^*,j}<\bar{c}^{i*}_j} f(i^*,j)$ 
        \State $rr_{i^*,j^*}\leftarrow rr_{i^*,j^*}-u_{i^*}$
        \State $rb_{j^*}\leftarrow rb_{j^*}-u_{i^*}$
        \State $x_{i^*,j^*}\leftarrow x_{i^*,j^*}+u_{i^*}$
        \If{$x_{i^*,\cdot}>(\bar{c}^{i^*}-1)\cdot u_{i^*}$}
            $rk_{i^*}\leftarrow +\infty$
        \Else 
            ~$rk_{i^*}\leftarrow x_{i^*,\cdot}/f_{i^*}$
        \EndIf
    \EndIf
\EndWhile
\State $p\leftarrow\min\{f_i/x_{i,\cdot}\}$
\end{algorithmic}
\end{algorithm}

\vspace*{-0.1in}
\section{Evaluations}

\subsection{Experimental Settings}

We consider a system with 
$s$ FL processes ongoing, 
$e$ edge servers, 
and up to $c$ clients for each FL process.
By default, $s=5$, and the FL servers are denoted as $S_0,\cdots,S_4$;
$e=5$ and the edge servers are denoted as  $E_0,\cdots,E_4$.
At each round of FL, 
each $S_i$ randomly picks $c$ clients where $c$ is set to 50 by default. 
Each $E_j$ has a limited bandwidth of $b_j = 10$ units.

The FL processes 
are based on the MNIST~\cite{lecun1998gradient} dataset. . 
We prepare the MNIST dataset in a non-iid manner regarding class label distribution. There are two distributions: first, each client has data samples of only 1 class label; second, each client has data samples of 2 class labels. The distribution of the data sample is equal to each client.
We use a simple CNN classifier with a PyTorch sequential model 
consisting of two convolutional layers and two FC layers. 
The model uses the ReLU activation function. 
The momentum is 0.5 and the learning rate is 0.01. 
Each concurrent FL process runs for 20 global rounds, 
each global round consists of 50 epochs.

In the experiments, 
we study the performance of our proposed scheme 
in the presence of a variety of heterogeneity. 

{\noindent \em \underline{Heterogeneity of Client Distribution}}:
To assess client distribution heterogeneity, 
we use parameters $\alpha$ and $\beta$ to 
control the heterogeneity of client distribution over edge servers.  
Specifically, when distributing 
$c$
clients of each $S_i$ to the edge servers:
    we pick $\alpha\cdot s$ FL servers to distribute their clients only over $\beta\cdot e$ edge servers,
    where $0\leq\alpha\leq1$ and $0\leq\beta\leq1$,
    while the rest $(1-\alpha)\cdot s$ FL servers uniformly randomly distribute their clients to all the $E$ edge servers. 
Note that the distribution is uniform when $\alpha=0$ and $\beta=1$.

{\noindent \em \underline{Heterogeneity of Funds}}:
Each FL server is set to have fund $f_i$
between $0.5$ and $1$.
  
That is, let $0\leq\gamma\leq 1$, and 
$\gamma\cdot s$ be the number of FL servers whose funds are between 
$0.5$ 
and $1$ 
while the rest $(1-\gamma)\cdot s$ FL servers each having the fund of 
$0.5$.
More specifically, 
the $\gamma\cdot s$ FL servers would have funds of 
$0.5 + \frac{k\cdot 0.5}{\gamma\cdot s}$,
where $k=1,\cdots, \gamma\cdot s$, respectively. 
Under such settings, we study how 
the performance of each FL server
is affected by its fund.

{\noindent \em \underline{Heterogeneity of Required Bandwidth Per Client}}: 
Each $S_i$ needs $u_i$ units of bandwidth for its clients, 
where $u_i$ ranges between $1$ and $s$.
Let $0\leq\delta\leq 1$, and $\delta\cdot s$ be the number of FL servers whose bandwidths are between $u_i=1$ and $s$, 
while the rest $(1-\delta)\cdot s$ FL servers, each having the bandwidths of $u_i$. 
More specifically, the $\delta\cdot s$ FL servers would have bandwidths of $u_i + \frac{k\cdot (s-u_i)}{\delta\cdot s}$, 
where $k=1,\cdots, \delta\cdot s$, respectively.

\subsection{Compared Baseline Scheme}

To assess the effectiveness of the proposed game-theoretic schemes, 
we conduct them with a baseline (distributed) scheme,
where each edge server allocates its bandwidth to requesting cloud servers proportionally to their requesting demand.

\subsection{Evaluation Metrics} 

In the evaluation work, 
we adopt the following metrics:
    (i) number of clients acquired by each FL server;
    (ii) Test accuracy achieved by each FL server.
Each FL server should try to acquire as much bandwidth as possible from the edge servers using its available fund,
to support as many of its clients as possible. 
With as many clients participating,
an FL server can achieve higher test accuracy for its learning.

\subsubsection{Number of Clients for Each FL Process} 

The FL servers will try to achieve the required up-link bandwidth from the edge server with available funds. The BandA greedy algorithm allocated the requesting up-link bandwidth to the FL server by each edge server with available bandwidth. The baseline schema allocated the requesting up-link bandwidth to the FL server by each edge server with available bandwidth proportionally to their requests. Here, $u_i = 1$ units of bandwidth are needed by each edge device. Therefore, the number of edge devices achieved by each FL server is a performance evaluation metric. We find the number of edge devices achieved by the FL server from the proposed game-based and baseline schema and then compare the performance of these two schemes. 

\subsubsection{Test Accuracy for Each FL Process} 

If an FL server achieved many edge devices from the BandA greedy algorithm, it could achieve much bandwidth and receive a large data sample from the selected edge devices. The locally trained model by the selected edge devices is aggregated by FL servers using FedAvg\cite{mcmahan2017communication}. The FL servers used their test data to infer the aggregated global model. Therefore, the test accuracy of the FL server is a performance metric.

\subsection{Evaluation Results I: Clients Acquired by FL Servers 
under a Variety of Heterogeneity}

\subsubsection{Client Distribution Heterogeneity}
\label{sec:client-distribution-heterogeneity}

We evaluate the proposed centralized and distributed game-theoretic scheme 
against the baseline (distributed) scheme by experimenting with 
different values of heterogeneity parameters $\alpha=[0.2, 0.4, 0.6, 0.8]$ and $\beta=[0.2, 0.4, 0.6, 0.8]$. 
Note that the distribution of the clients is uniform when $\alpha=0$ and $\beta=1$. 

From  Figures~\ref{fig:bar_base_client_hetero},
\ref{fig:bar_game_client_hetero} and \ref{fig:distr_bar_game_client_hetero}, 
the proposed game-theoretic schemes allocate bandwidth more fairly 
then the baseline scheme, given 
the heterogeneous distribution of clients 
(and thus heterogeneous bandwidth requests) over the edge servers. 
That is, all five FL servers acquire an equal number of clients. 
For example, when $\alpha=0.4$ and $\beta=0.6$, 
edge servers $E_0$, $E_1$ and $E_2$ receive more bandwidth requests 
than $E_3$ and $E_4$. Using the baseline scheme, 
$S_0$ and $S_1$ acquire less bandwidth compared to $S_2$, $S_3$ and $S_4$ 
because they receive bandwidth proportionally to their requests.

In the extreme scenario of resource distribution, 
the value of heterogeneity parameter $\alpha$ becomes much greater than $\beta$, which means a significantly larger number of FL servers sending bandwidth requests 
to a smaller number of edge servers. There may not be enough bandwidth 
at the edge servers to assign bandwidth fairly to the requesting servers. 
Even proposed game-theoretic schemes may not fairly assign bandwidth
to the FL server. However, proposed schemes perform well except for such scenarios.

\begin{figure}
    \centering
    \includegraphics[width=8.5cm]{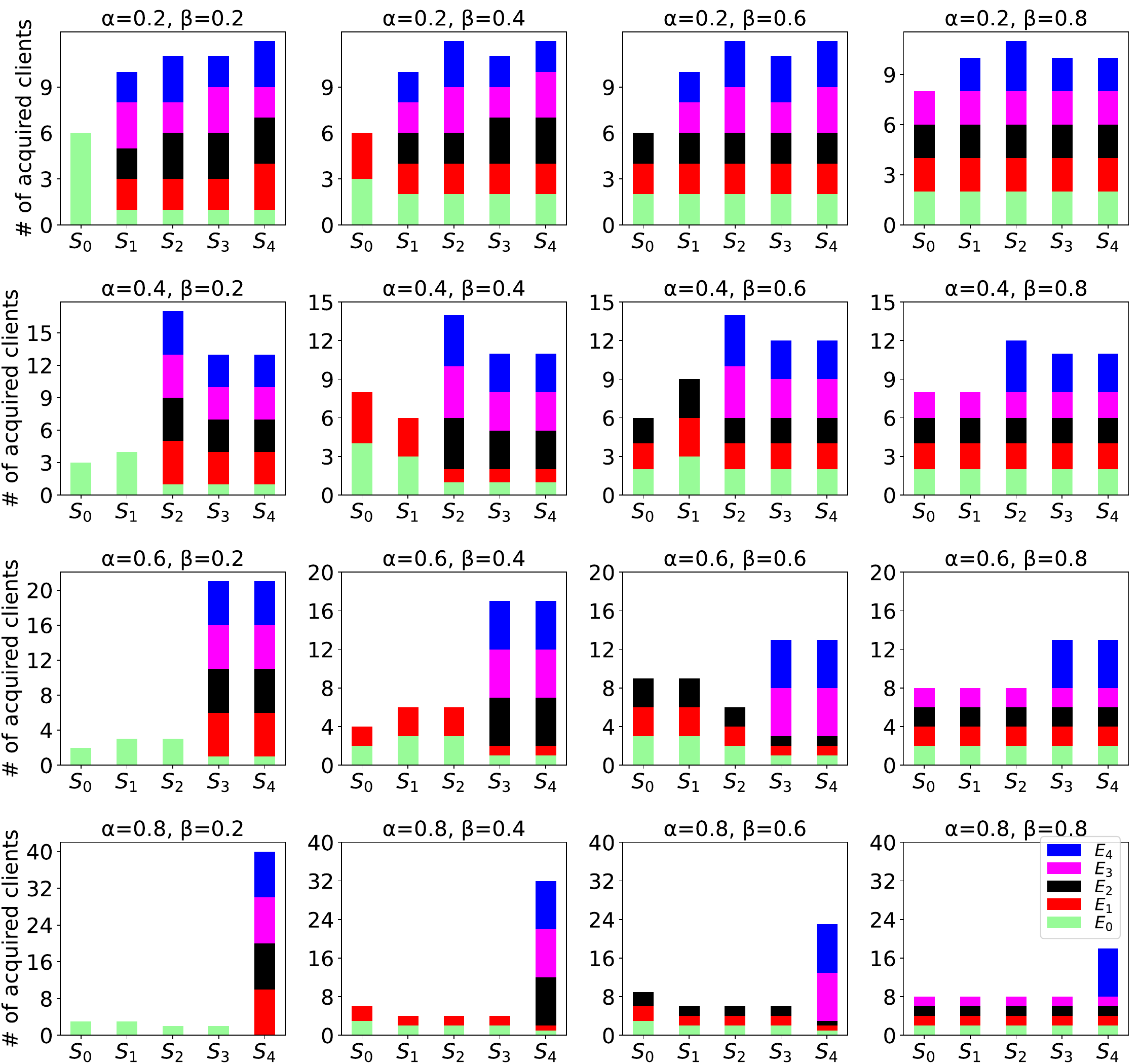}
    \caption{Numbers of clients acquired by FL servers with Client Heterogeneity (Baseline Scheme).}
    \label{fig:bar_base_client_hetero}
\end{figure}
\begin{figure}
    \centering
    \includegraphics[width=8.5cm]{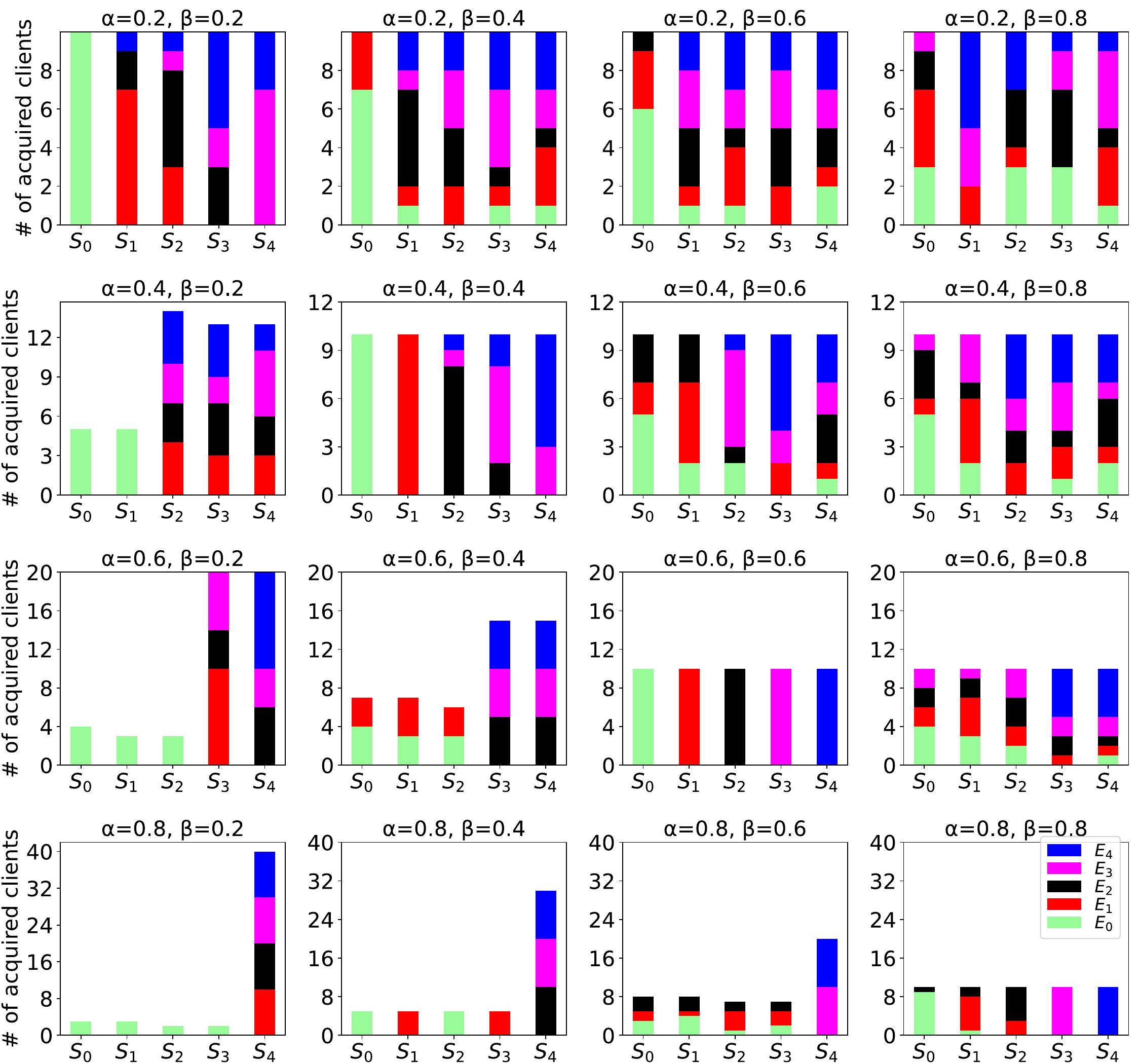}
    \caption{Numbers of clients acquired by FL servers with Client Heterogeneity (Proposed Centralized Scheme).}
    \label{fig:bar_game_client_hetero}
\end{figure}
\begin{figure}
    \centering
    \includegraphics[width=8.5cm]{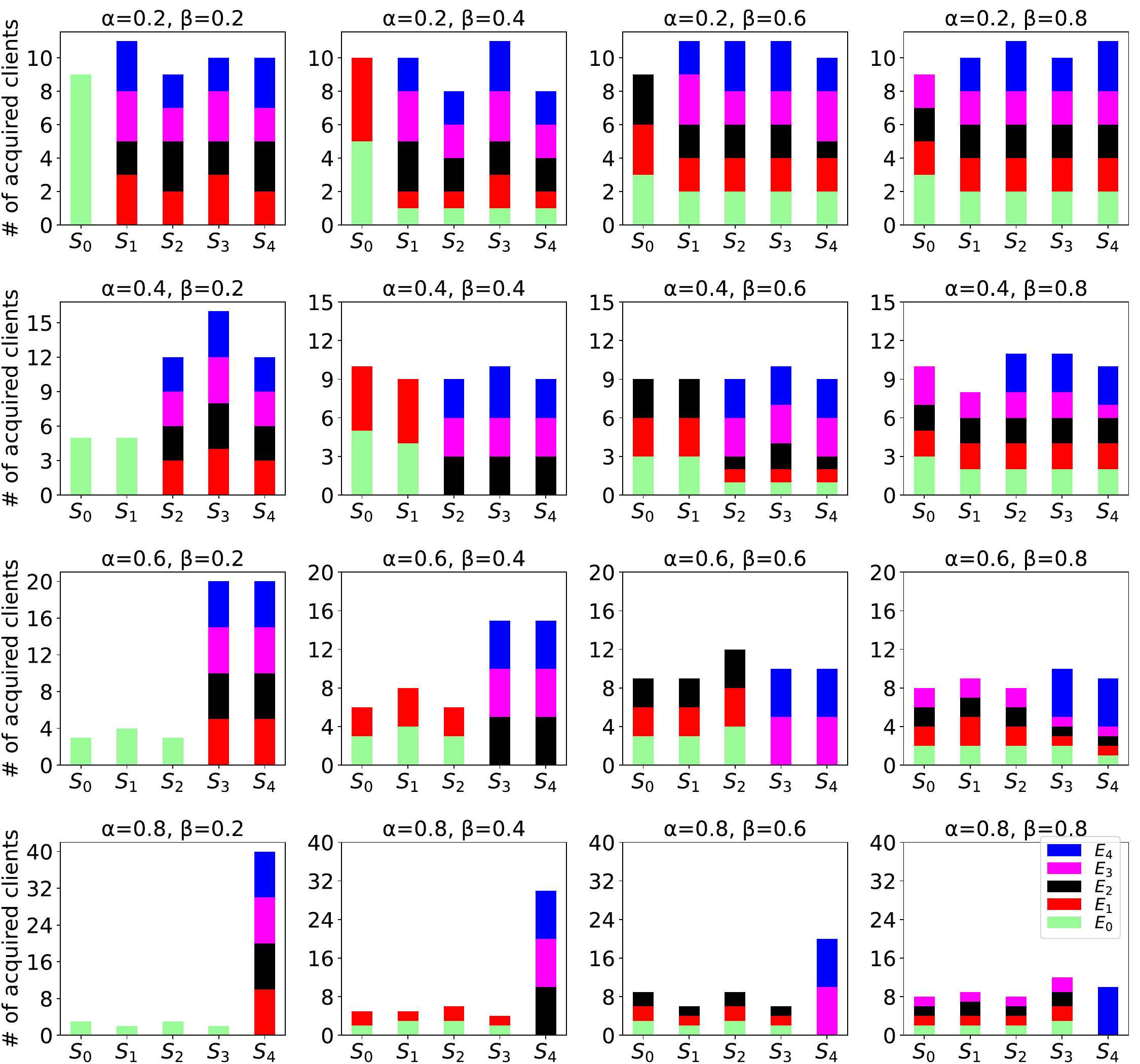}
    \caption{Numbers of clients acquired by FL servers with Client Heterogeneity (Proposed Distributed Scheme).}
    \label{fig:distr_bar_game_client_hetero}
\end{figure}

\subsubsection{Fund heterogeneity}
We vary fund heterogeneity parameter $\gamma$ between $0$ and $0.8$, which results in the heterogeneous distribution of funds among the FL servers. 
In Figures~\ref{fig:fund_hetero_ed_es} (a) and~\ref{fig:fund_hetero_ed_es} (b), the FL servers can acquire roughly an equal number of clients when they have the same fund ($\gamma = 0$) using proposed centralized and distributed game-theoretic scheme. When FL server's funds differ ($\gamma = [0.4, 0.6, 0.8])$, an FL server with a larger fund can acquire more clients.

\begin{figure}
    \centering
    \includegraphics[width=8.5cm]{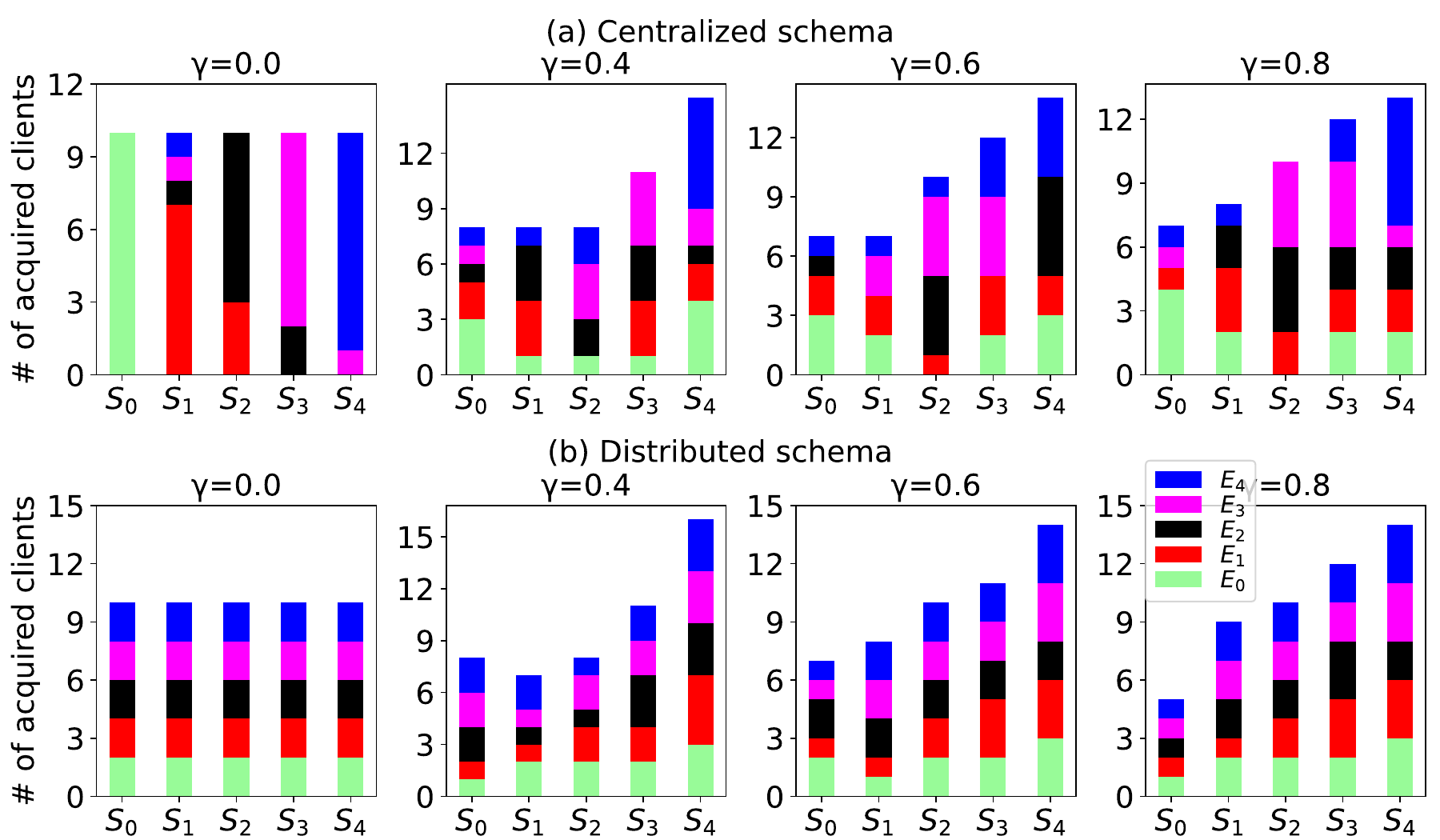}
    \caption{Numbers of clients acquired by FL servers (Heterogeneity of Funds).}
    \label{fig:fund_hetero_ed_es}
\end{figure}

\subsubsection{Heterogeneity of Required Bandwidth Per Client}
\begin{figure}
    \centering
    \includegraphics[width=8.5cm]{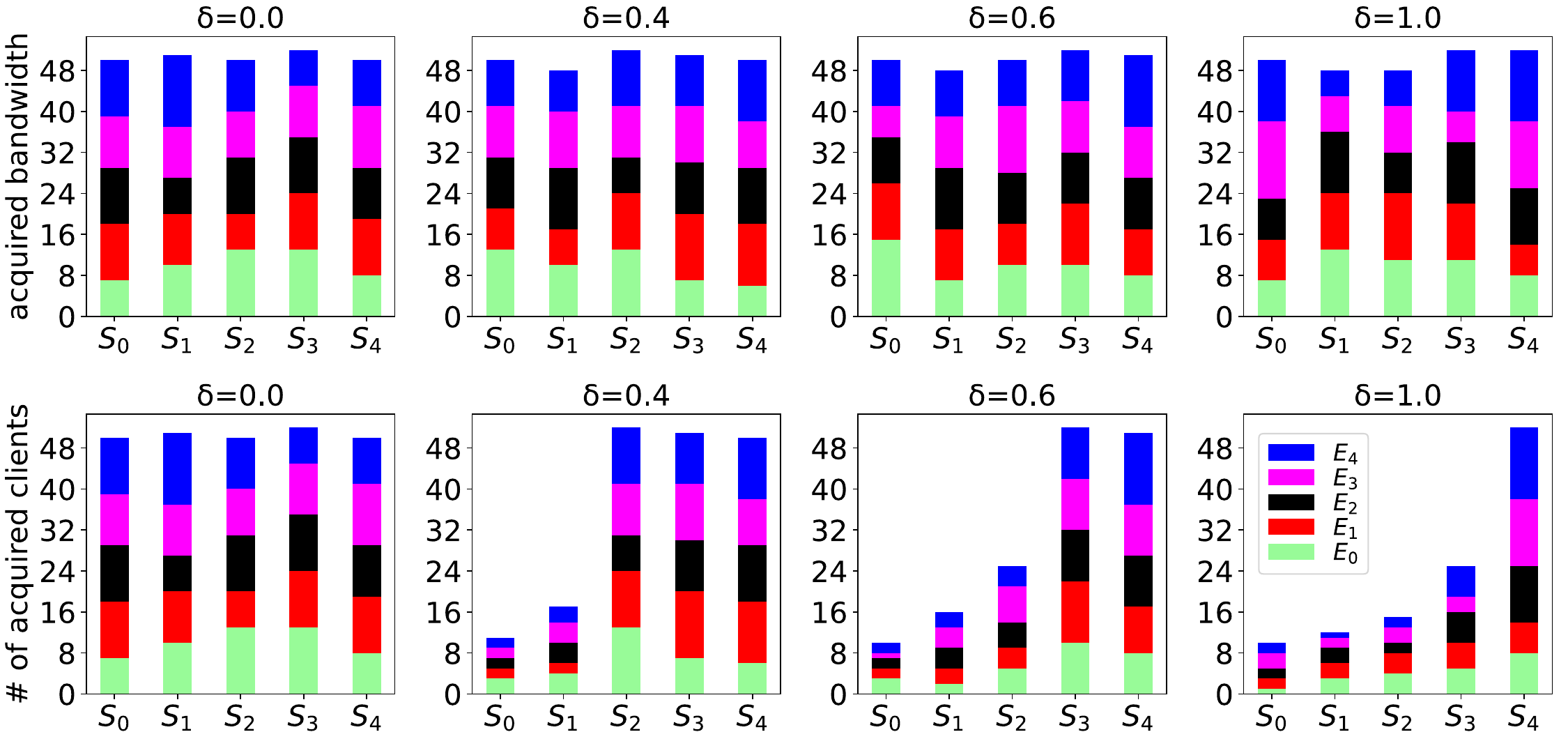}
    \caption{Bandwidth and number of clients acquired by FL servers (Heterogeneity of Required Bandwidth Per Client - Centralized Scheme).}
    \label{fig:bandwidth_hetero}
\end{figure}

\begin{figure}
    \centering
    \includegraphics[width=8.5cm]{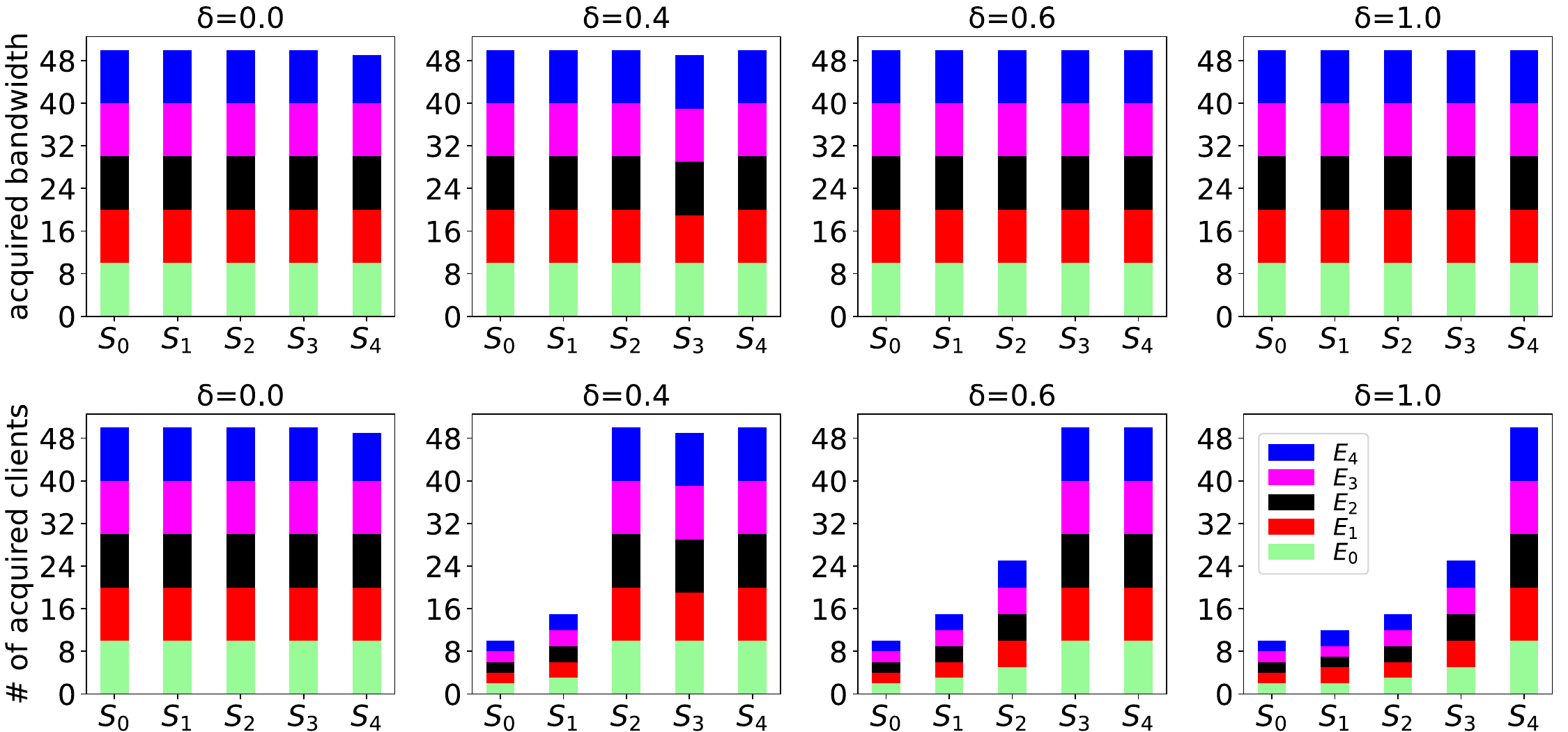}
    \caption{Bandwidth and number of clients acquired by FL servers (Heterogeneity of Required Bandwidth Per Client - Distributed Scheme.)}
    \label{fig:fig:bandwidth_hetero_distr}
\end{figure}
We vary $\delta$ between $0$ and $1$, 
which results in that 
different FL servers may require different amounts of bandwidth per client. We consider client $c = 250$ and $b_j = 50$ only for this heterogeneity.
In Figures~\ref{fig:bandwidth_hetero} and \ref{fig:fig:bandwidth_hetero_distr}, all FL servers are allocated with a similar amount of bandwidth, as proposed game-theoretic schemes allocate bandwidth fairly to FL servers. Therefore, the number of clients they acquire differs due to the difference in the required bandwidth per client. According to Figures~\ref{fig:bandwidth_hetero} and \ref{fig:fig:bandwidth_hetero_distr}, a larger bandwidth required per client, the fewer clients are acquired. 

\begin{table*}
\caption{Test Accuracy with Client Heterogeneity}
\begin{center}
\begin{tabular}{|cc|c|c|c|c|c|c|c|c|c|c|c|c|l}
\cline{1-14}
&  &  \multicolumn{6}{ c| }{Each Client having Training Data with 1 class label} & \multicolumn{6}{ c| }{Each Client having Training Data with 2 class labels} \\ \cline{1-14}
&  & \multicolumn{2}{ c| }{Acc(\%) Baseline} & \multicolumn{2}{ c| }{Acc(\%) Centralized} & \multicolumn{2}{ c| }{Acc(\%) Distributed} & \multicolumn{2}{ c| }{Acc(\%) Baseline} & \multicolumn{2}{ c| }{Acc(\%) Centralized} & \multicolumn{2}{ c| }{Acc(\%) Distributed}  \\ \cline{1-14}
$\alpha$ & $\beta$ & $S_0$ & $S_4$ & $S_0$ & $S_4$ & $S_0$ & $S_4$ & $S_0$ & $S_4$ & $S_0$ & $S_4$ & $S_0$ & $S_4$ \\ \cline{1-14}
\multicolumn{1}{ |c  }{\multirow{4}{*}{0.2} } &
\multicolumn{1}{ |c| }{0.2} & 34.00  & 39.12 & 41.76  & 38.67 & 35.39 & 41.01 &  67.04  &  80.16 & 77.35  & 76.25 & 77.73 & 76.35 & \\ \cline{2-14}
\multicolumn{1}{ |c  }{}                        &
\multicolumn{1}{ |c| }{0.4} & 36.97 & 38.76 & 41.26  & 39.05 & 41.61 & 42.90 &  68.33 & 78.60 & 78.69   & 78.55 & 79.69 & 78.21 &      \\ \cline{2-14}
\multicolumn{1}{ |c  }{}                        &
\multicolumn{1}{ |c| }{0.6} & 36.65  &  39.63 & 38.05 &  39.36 & 41.05 & 41.11 &  69.92 &  81.39 & 80.41   & 83.03  & 74.28 & 78.57 &    \\ \cline{2-14}
\multicolumn{1}{ |c  }{}                        &
\multicolumn{1}{ |c| }{0.8} & 37.29 & 39.40 & 41.01 & 39.99 &  36.90 & 40.61 & 80.69   & 82.09 & 79.53  & 80.14 & 78.35 & 82.88 &    \\ \cline{1-14}
\multicolumn{1}{ |c  }{\multirow{4}{*}{0.4} } &
\multicolumn{1}{ |c| }{0.2} & 20.24  &  37.99 & 30.98   & 40.70 & 29.50 & 42.00 &  51.27 &  81.23 & 58.86   & 83.29 & 63.25  & 83.27 &   \\ \cline{2-14}
\multicolumn{1}{ |c  }{}                        &
\multicolumn{1}{ |c| }{0.4} & 36.66 &  39.45 & 41.14  & 37.32  & 37.13 & 41.17 &   73.21 &  81.57 & 76.60   & 80.42  & 80.23 & 81.91 &   \\ \cline{2-14}
\multicolumn{1}{ |c  }{}                        &
\multicolumn{1}{ |c| }{0.6} & 37.44  &  38.10 & 42.21  & 38.54 & 40.69 &  41.32 &   69.86  & 82.15 & 79.61 & 79.75 & 74.12 & 80.18 &     \\ \cline{2-14}
\multicolumn{1}{ |c  }{}                        &
\multicolumn{1}{ |c| }{0.8} & 37.09  & 42.15 & 38.42  & 38.68  &  37.04 & 40.15 &  74.79 &  82.11 & 80.99  & 79.4  & 79.87  & 74.49 &    \\ \cline{1-14}
\multicolumn{1}{ |c  }{\multirow{4}{*}{0.6} } &
\multicolumn{1}{ |c| }{0.2} &  18.83 &  42.45 & 27.74  & 44.31  & 23.35 &  44.13 & 32.16  &  87.84 & 58.70  & 87.76  & 48.91 & 87.43 &    \\ \cline{2-14}
\multicolumn{1}{ |c  }{}                        &
\multicolumn{1}{ |c| }{0.4} & 26.60 &  43.25 & 35.39  & 43.46  & 34.74 & 39.14 &  56.35 & 87.97 & 72.66 &  83.55  & 69.25 & 82.25 &    \\ \cline{2-14}
\multicolumn{1}{ |c  }{}                        &
\multicolumn{1}{ |c| }{0.6} & 37.15 &  45.17 & 43.12  & 43.53  & 38.75 & 45.20 &  79.17 &  84.41 & 78.46 & 77.63  &  82.06 & 79.98 &   \\ \cline{2-14}
\multicolumn{1}{ |c  }{}                        &
\multicolumn{1}{ |c| }{0.8} & 35.50  &  41.68 & 40.52  & 44.08 &  39.43 & 39.98 & 74.43 &  82.43 & 82.25  &  77.53  & 78.04 & 77.72 &   \\ \cline{1-14}
\multicolumn{1}{ |c  }{\multirow{4}{*}{0.8} } &
\multicolumn{1}{ |c| }{0.2} & 24.18 &  46.64 & 22.30 & 49.51 &  21.58 & 51.48 &   47.66 & 92.12 & 47.72  & 92.81 & 50.80  & 92.79 &    \\ \cline{2-14}
\multicolumn{1}{ |c  }{}                        &
\multicolumn{1}{ |c| }{0.4} & 31.74 &  48.23 & 30.98  & 48.19 &  33.04 & 44.72 & 69.33 & 91.79 & 67.03  & 91.79 & 65.04  & 91.83 &    \\ \cline{2-14}
\multicolumn{1}{ |c  }{}                        &
\multicolumn{1}{ |c| }{0.6} & 37.52 &  46.15 & 35.60  & 42.98  & 40.99 & 42.97 &  77.24 &  90.37 & 80.64  & 88.29  & 79.43  & 87.42 &   \\ \cline{2-14}
\multicolumn{1}{ |c  }{}                        &
\multicolumn{1}{ |c| }{0.8} & 35.30 &  41.88 & 38.52 & 37.13 &   38.93 & 41.14 &  80.94 & 86.98 & 80.29 & 81.81 & 83.24 & 80.85 &   \\ \cline{1-14}
\multicolumn{1}{ |c  }{\multirow{1}{*}{0.0} } &

\multicolumn{1}{ |c| }{1.0} &  40.10 &  40.00 & 38.65 & 40.04 &  36.88 & 39.73 &  80.92 &  82.26 & 81.08  & 79.07 & 77.95  & 79.14 &   \\ \cline{1-14}
\end{tabular}
\label{tbl:distr_fs_test_accu_client_hetero_mnist_class1}
\end{center}
\end{table*}

\subsection{Evaluation Results II: Test Accuracy in FL Servers}

Next, we measure the FL servers' test accuracy in the afore-discussed heterogeneous scenarios.

\subsubsection{Client Distribution Heterogeneity}
Following the same settings in 
Section~\ref{sec:client-distribution-heterogeneity},
we evaluate and compare the performance of 
the proposed game-theoretic schemes against the baseline scheme
in terms of test accuracy at FL servers. In a uniform client distribution scenario, $\alpha=0.0$ and $\beta=1.0$, the baseline and game base models perform similar accuracy shown table \ref{tbl:distr_fs_test_accu_client_hetero_mnist_class1}. FL servers achieved similar test accuracy in both models on 1 class label non-iid data distribution. The FL server's test accuracy is also close enough on 2 class labels data distribution.
First, when $\alpha=0.0$ and $\beta=1.0$ 
(i.e., clients are uniformly distributed to edge servers),
the baseline and the game-theoretic schemes achieve similar test accuracy,
as shown in the last rows of 
Tables~\ref{tbl:distr_fs_test_accu_client_hetero_mnist_class1}. 
In the first experiment, 
each client has training data with only 1 class label; 
in another experiment, 
each client has training data with 2 class labels. 
Therefore, not all FL servers can obtain models 
trained from the data with most class labels.
Particularly, 
if an FL server acquires a small number of clients, 
its model is trained based on the data with only a few class labels.
Also, since 
clients are randomly selected from the client pool, 
an FL server could receive a model updated based on the data with 
different number of class labels 
in different rounds. 
Thus, testing accuracy varies in rounds, 
even though they get model updates from the same number of clients.
As we can see, the baseline scheme does not assign bandwidth fairly. 
For example, when $\alpha=0.4$ and $\beta=0.6$, 
FL server $S_0$ acquires less bandwidth than $S_4$;
therefore, it has lower test accuracy than $S_4$. 
In the game-theoretic schemes, 
except for the highly heterogeneous scenarios such as when
($\alpha=0.6,\beta=0.2$), ($\alpha=0.8,\beta=0.2$) and ($\alpha=0.8,\beta=0.4$), 
we can see that FL servers
$S_0$ and $S_4$ have similar test accuracy and are better than baseline schema.

\subsubsection{Fund Heterogeneity}

We evaluate the impact of fund heterogeneity.
Table~\ref{tbl:fund_heterogeneity_acc} illustrates the test accuracy 
as heterogeneity parameter $\gamma$ varies. We can observe that an FL server achieves higher test accuracy when it has a larger fund. For example, in the distributed scheme and training data with 2 class labels, for $\gamma=0.4$, FL server $S_0$ has a fund of 0.5 and a test accuracy of $79.05\%$, whereas  $S_4$ has a larger fund of 1.0 and a higher test accuracy of $86.71\%$.  Therefore, $S_4$ acquired more bandwidth with the fund and achieved higher test accuracy.

\begin{table}
\caption{Test Accuracy with Fund Heterogeneity}
\begin{center}
\begin{tabular}{|c|c|c|c|c|c|c|l}
\cline{1-7}

& \multicolumn{6}{ c| }{Game-theoretic Centralized Scheme}  \\ \cline{1-7}
& \multicolumn{3}{ c| }{Acc(\%) 1 class label} & \multicolumn{3}{ c| }{
Acc(\%) 2 class labels}  \\ \cline{1-7}
$\gamma$ & $S_0$ & $S_2$ & $S_4$ & $S_0$ & $S_2$ & $S_4$ \\ \hline
0.0 & 41.40 & 37.73 & 38.76 & 78.43 & 79.28  & 79.67 \\ \hline
0.4 & 39.24 & 35.77 & 39.42 & 77.01 & 77.61  & 87.16 \\   \hline       
0.6 & 35.92 & 37.95 & 42.23 & 72.70 & 79.67  & 81.09 \\  \hline        
0.8 & 34.26 & 39.97 & 42.77 & 75.94 & 82.08  & 85.30 \\   \hline
& \multicolumn{6}{ c| }{Game-theoretic Distributed Scheme} \\ \hline
0.0 & 40.07 & 41.55 & 42.76 & 80.80 & 80.26 & 80.00   \\ \hline
0.4 & 36.22 & 35.15 & 46.78 & 79.05 & 74.05 & 86.71   \\ \hline
0.6 & 33.57 & 43.00 & 37.37 & 65.37 & 77.46 & 85.67   \\ \hline
0.8 & 31.43 & 36.68 & 43.78 & 66.33 & 78.41 & 81.91  \\ \hline
\end{tabular}
\label{tbl:fund_heterogeneity_acc}
\end{center}
\end{table}

\subsubsection{Impact of $\rho$}
Recall that Figure~\ref{fig:bar_game_client_hetero} shows 
an FL server receives different numbers of clients 
with different parameter values $\rho$. For the proposed game-theoretic scheme, table~\ref{tbl:user_fraction} shows that an FL server's test accuracy is higher when it receives model updates from a large number of clients. 
For example, when $\rho = 0.6$, 
FL server $S_0$ selects $60\%$ of all feasible clients in each training round, and 
achieves a test accuracy of $70.50\%$. 
On the other hand, when $\rho = 1.0$, 
FL server $S_0$ selects all of its feasible clients in each training round, 
and achieves a test accuracy of $79.34\%$.

\begin{table}[H]
\caption{Test Accuracy with Varying $\rho$}
\begin{center}
\begin{tabular}{|c|c|c|c|c|c|c|l}
\cline{1-7}
& \multicolumn{3}{ c| }{Acc(\%) 1 class label} & \multicolumn{3}{ c| }{
Acc(\%) 2 class labels}  \\ \cline{1-7}
$\rho$ & $S_0$ & $S_2$ & $S_4$ & $S_0$ & $S_2$ & $S_4$ \\ \hline
0.4 & 28.93 & 23.12 & 22.81 & 52.38 & 48.14 & 47.27 \\ \hline
0.6 & 34.58 & 29.93 & 31.57 & 70.50 & 61.86 & 62.89 \\  \hline          
0.8 & 39.41 & 33.78 & 35.00 & 76.77 & 66.91 & 70.45 \\  \hline           
1.0 & 39.53 & 37.08 & 40.00 & 79.34 & 78.33 & 79.99 \\ \hline    
\end{tabular}
\label{tbl:user_fraction}
\end{center}
\end{table}

\vspace*{-0.1in}

\section{Related Works}
\label{sec:related}

Since the inception of federated learning~\cite{Jakub45648}, 
research has been ongoing to enhance 
the efficiency~\cite{wang2018edge,10.1109/WCNC49053.2021.9417299,li2019smartpc,liu2020accelerating}, 
privacy~\cite{wang2019beyond,geyer2017differentially}, 
security~\cite{blanchard2017machine,fang2020local}, 
incentives~\cite{jiao2020toward,zhan2020learning}, 
among others.

Communication efficiency~\cite{li2020federated} in FL 
has been identified as a significant hurdle. 

Research focuses on adapting the FL algorithm 
to mitigate the communication load on the network, 
involving techniques like 
selecting clients that demonstrate notable training advancements, 
compressing gradient vectors through quantization, or 
speeding up training through sparse or structured updates. 
Works show the development of new optimization and model aggregation algorithms such as \cite{chen2018lag,lin2017deep,Jakub45648,aji2017sparse,karimireddy2019scaffold,li2020federated,haddadpour2019convergence}. 
In \cite{pmlr-v162-yi22a}, the authors proposed a QSFL framework that optimizes FL uplink communication overhead from two levels: 
rejecting redundant clients uploading model updates and 
compressing models to unique segments for uploading. 
Studies are also conducted to reduce communication costs 
by shifting extensive model aggregations to the edge. 
This approach avoids substantial communication overhead to the cloud 
while maintaining learning performance. 
Consequently, studies have focused on 
the Hierarchical Federated Learning (HFL) framework and explored communication efficiency issues.
A concept of Hierarchical Federated Learning networks has been put forth in \cite{liu2020client}, 
where several edge servers initially conduct partial model aggregation. 
A cloud server subsequently combines these intermediate results to generate the final model. 
In \cite{luo2020hfel},
they formulate a joint computation and communication resource allocation and edge association problem 
for device users under Hierarchical Federated Edge Learning (HFEL) framework 
to achieve global cost minimization. 

Wang et al. proposed an approach in \cite{wang2021resource} that uses clusters and hierarchical aggregation. 
The authors introduced an effective algorithm to determine the optimal number of clusters while considering resource limitations, 
and they carried out the training using edge computing. 
Yang et al. introduced an architecture in \cite{yang2021h} for Hierarchical Federated Learning (H-FL). 
 
The authors develop a strategy to reconstruct the runtime distribution to counteract the degradation caused by non-IID data. 
Additionally, they design a mechanism for compression correction that reduces communication overhead without compromising model performance. 
Deng et al.~\cite{deng2021share} formulated the communication cost minimization (CCM) problem 
to decrease the communication costs that arise from combining models from the edge and cloud within the Hierarchical Federated Learning (HFL) framework. 
The authors suggest that 
improving the uniformity of the training data at the edge aggregator 
can greatly enhance the performance of Federated Learning (FL). 
They propose SHARE, which stands for SHaping dAta distRibution at Edge, 
to efficiently solve the CCM problem by modifying the data distribution at the edge. 
Several studies~\cite{kang2019incentive,joint,zou2019mobile} 
have applied the game theory to FL systems.

Liu et al.~\cite{liu2020client} proposed an incentive mechanism based on contract theory 
to encourage data owners with high-quality local training data to participate in the learning processes. 
Feng et al.~\cite{joint} used a Stackelberg game model 
to analyze self-organized mobile devices' transmission and training data pricing strategies and the model owner's learning service subscription in the cooperative federated learning system.  
Zou et al.~\cite{zou2019mobile} employed an evolutionary game theory to dynamically model the strategies of mobile devices 
with bounded rationality in the federated learning system.
As far as we know, none of the previous works have addressed the scenarios where 
multiple federated servers run simultaneously, 
with edge servers serving as intermediaries between FL servers and edge devices. 
Our paper tackles such scenarios with a game-theoretic approach.

\section{Conclusion}
\vspace{-0.1in}

This paper studied the bandwidth allocation problem in 
a three-tier FL system with multiple FL processes run concurrently.
We proposed a Stackelberg game-theoretic approach to formulate the problem and
designed heuristic distributed and centralized schemes to find an approximate Nash Equilibrium of the game.
We have conducted extensive experiments, and the results indicate that
our proposed scheme can efficiently and fairly allocate bandwidth at edge servers to support concurrent FL processes.

\end{document}